\documentclass[12pt]{amsart}

\usepackage[matrix,arrow,curve,frame]{xy}
\usepackage{amsmath,amsthm,amssymb,enumerate}
\usepackage{latexsym}
\usepackage{amscd}
\usepackage{euscript}
\usepackage[]{inputenc}
\usepackage{tikz}
\usepackage{tikz-cd}

\usepackage{fullpage}

\usepackage[all]{xypic}

\newtheorem{theorem}{Theorem}[section]
\newtheorem{lemma}[theorem]{Lemma}

\theoremstyle{definition}
\newtheorem{definition}[theorem]{Definition}

\theoremstyle{remark}
\newtheorem{remark}[theorem]{Remark}

\numberwithin{equation}{section}
\newcommand{\beq}{\begin{equation}}
\newcommand{\eeq}{\end{equation}}

\newcommand{\TT}{\mathbb{T}}
\newcommand{\ZZ}{\mathbb{Z}}
\newcommand{\RR}{\mathbb{R}}
\newcommand{\CC}{\mathbb{C}}
\newcommand{\HH}{\mathbb{H}}

\newcommand\GG{\mathrm{G}}
\newcommand\GGG{\mathcal{G}}

\newcommand{\sU}{{\sf U}}

\newcommand{\llz}{C^\infty(T^2, Z)}
\newcommand{\llhatz}{C^\infty(T^2, \widehat Z)}

\newcommand{\cS}{\mathcal{S}}

\newcommand{\cO}{\mathcal{O}}

\renewcommand{\cL}{\mathcal{L}}

\newcommand{\gH}{\mathfrak{H}}
\newcommand{\gI}{\mathfrak{I}}

\newcommand{\ktau}{{K_1+\tau K_2}}

\newcommand{\bu}{\bullet}
\newcommand{\bH}{\mathbb{H}}
\newcommand{\gch}{\mathrm{GCh}_H}

\newcommand {\be}{\begin{equation}}
\newcommand {\ee}{\end{equation}}
\newcommand{\h}{\begin{eqnarray*}}
\newcommand{\e}{\end{eqnarray*}}

\begin{document}


\title[T-duality with $H$-flux for $2d$ $\sigma$-models]
{T-duality with $H$-flux for $2d$ $\sigma$-models}


\author{Fei Han}
\address{Department of Mathematics,
National University of Singapore, Singapore 119076}
\email{mathanf@nus.edu.sg}

 \author{Varghese Mathai}
\address{School of Mathematical Sciences,
University of Adelaide, Adelaide 5005, Australia}
\email{mathai.varghese@adelaide.edu.au}

\date{}

\maketitle

\begin{abstract} In this paper, we establish graded T-duality for $2d$ $\sigma$-models with $H$-flux after localization. 
 This establishes the most general version of T-duality for Type II String Theory. 
The graded T-duality map, which we call {\bf graded Hori morphism}, is compatible with the Jacobi property of the graded fields, that was earlier studied in 
\cite{HM21}. Also included are some open problems/conjectures.
\end{abstract}

\tableofcontents


\section*{Introduction}

T-duality in string theory can be realised as a transformation acting on the
worldsheet fields in the two-dimensional nonlinear sigma model \cite{HLSUZ}.
In \cite{KL, HLSUZ2}, T-duality is studied with supersymmetry in the target space of the sigma model.  This
realisation is straightforward when there is an abelian isometry of the target
space which at the same time is a symmetry of the worldsheet action.
One finds the transformation on the worldsheet fields as
well as Buscher's rules \cite{Buscher} which yield the dual
background.
The duality transformation is symmetric in the sense that
one can start from either theory (the original or the dual one) and
obtain the other one via the same gauging process \cite{Alvarez2}. A relevant account of the sigma model is in \cite{W88}.

This paper proves T-duality for {two-dimensional} sigma models on circle bundles with $H$-flux, after localization. To be more precise, given a T-dual pair (see Section \ref{revT} for details)
\begin{center}
\begin{tikzcd}
(Z, A, H)\arrow[rd, "\pi"] & & (\widehat{Z}, \widehat A, \widehat H) \arrow[ld, "\widehat{\pi}"'] \\
 & X & 
\end{tikzcd}
\end{center}
it was shown in \cite{BEM04a, BEM04b} that the Hori map
\begin{align*}
T: (\Omega^{\bullet}(Z)^{\TT}, d+H) &\to (\Omega^{\bullet+1}(\widehat{Z})^{\widehat \TT}, -(d+{\widehat{H})})\\
\omega \quad &\mapsto \int_{\TT}\omega \wedge e^{-\widehat{A} \wedge A},
\end{align*}
is a chain map isomorphism between the twisted, $\mathbb{Z}_2$-graded complexes. One can also define the dual Hori map
\begin{align*}
\widehat T: (\Omega^{\bullet}(\widehat Z)^{\widehat \TT}, d+\widehat H) &\to (\Omega^{\bullet+1}(Z)^{\TT}, -(d+H))\\
\omega \quad &\mapsto \int_{\widehat \TT}\omega \wedge e^{-\widehat{A} \wedge A},
\end{align*} and T-duality can be stated as
\be \label{tdual} T\circ \widehat T=-Id, \ \ \widehat T\circ T=-Id.\ee In particular, this induces an isomorphism on the twisted cohomology:
\begin{align*}
T: H^{\bullet}_{d+H}(Z) &\to H^{\bullet+1}_{d+\widehat{H}}(\widehat{Z}).
\end{align*}
The Chern class is exchanged with the $H$-flux. So in general the topologies of $Z$ and the T-dual $\widehat Z$ are different.
 
{To prove T-duality with $H$-flux in the $2d$ sigma models, we will develop the following theory.} Let $M$ be a smooth manifold carrying a flux $H$. Let $C^\infty(T^2, M)$ be the double loop space of $M$, namely the space of smooth maps from $T^2=S^1\times S^1$ to $M$. Let $K_1, K_2$ be the two commuting vector fields on $C^\infty(T^2, M)$ obtained by rotating the first circle and second circle respectively. Let $\HH$ be the upper half plane. For each $\tau\in \HH$, we will construct a line bundle $\cL_{H, \tau}$ with a connection $\nabla^{\cL_{H, \tau}}$ on $C^\infty(T^2, M)^H$, the total space of a principal circle bundle over the double loop space $C^\infty(T^2, M)$ natually arising from the $H$-flux (see the definitions in Section \ref{construction}). We call $\cL_{H, \tau}$ the {\bf average $\tau$-holonomy line bundle}. This bundle shall be viewed as a double loop space complexification of the holonomy line bundle $\cL_H$ on $LM$ arising from a gerbe with connection whose curvature is $H$ (\cite{Bry}, c.f. \cite{HM15}) by taking into account of the complex structure $\tau$ on the space time $T^2$ (See Section \ref{explanation} for details).

Now consider $\Omega^{\bullet}_{bas}(C^\infty(T^2, M)^H, \cL_{H, \tau})[[u, u^{-1}]]$, where $\Omega^{\bullet}_{bas}(C^\infty(T^2, M)^H, \cL_{H, \tau})$ denotes the space of basic differential forms on $C^\infty(T^2, M)^H$ with values in the average $\tau$-holonomy line bundle $\cL_{H, \tau}$ and $u$ is an indeterminate of degree 2. {We will show that there is an odd operator $Q_{H, \tau}$ acting on $\Omega^{\bullet}_{bas}(C^\infty(T^2, M)^H, \cL_{H, \tau})[[u, u^{-1}]]$ such that (Theorem \ref{flat})
\be \frac{1}{2}[ Q_{H, \tau}, Q_{H, \tau}]=Q_{H, \tau}^2=-uL_{K_1+\tau K_2}^{\cL_{H, \tau}}, \ \ [Q_{H, \tau}, -uL_{K_1+\tau K_2}^{\cL_{H, \tau}}]=0,\ee
where $L_{K_1+\tau K_2}^{\cL_{H, \tau}}$ is the Lie derivative along the direction $K_1+\tau K_2$. So the odd operator $Q=Q_{H, \tau}$ and the even operator $P=-uL_{K_1+\tau K_2}^{\cL_{H, \tau}}$ obey the relations 
$$\frac{1}{2}[Q, Q]=P, \ \ [Q, P]=0$$
of the superalgebra considered in Witten's celebrated paper \cite{W82}.} Next we will prove that the restriction map gives a quasi-isomorphism between complexes (Theorem \ref{localization})
\be \label{BWloc}
\left(\Omega^\bullet_{bas}(C^\infty(T^2, M)^H, \cL_{H, \tau})^{K_1+\tau K_2}[[u, u^{-1}]], Q_{H, \tau}\right) \mapsto \left(\Omega^\bullet(M)[[u, u^{-1}]], d+u^{-1} H\right).
\ee
We would like to emphasize that such a Borel-Witten type localization theorem does not hold when $\tau$ are reals. Our method to prove the localization is motivated by \cite{JP} as well as our earlier work \cite{HM15}. See Section \ref{localisation} for details. 

Applying the theory to the situation of T-duality, we get the following picture:
\begin{center}
\begin{tikzcd}
(\cL_{H, \tau}, \nabla^{\cL_{H, \tau}})\arrow[d] & &(\cL_{\widehat H, \tau}, \nabla^{\cL_{\widehat H, \tau}})\arrow[d]  \\
C^\infty(T^2, Z)^H\arrow[rd, "LL\pi"] & & C^\infty(T^2, \widehat Z)^{\widehat H}\arrow[ld, "LL\widehat{\pi}"'] \\
 & C^\infty(T^2, X) & 
\end{tikzcd}
\end{center}
With respect to the $\TT$-action on $Z$ and $\widehat \TT$-action on $\widehat Z$ respectively, we deduce that
$$
\left(\Omega^{\bullet, \TT}_{bas}(C^\infty(T^2, Z)^H, \cL_{H, \tau})^{K_1+\tau K_2}[[u, u^{-1}]], Q_{H, \tau}\right) 
$$
and 
$$
\left(\Omega^{\bullet+1, \widehat \TT}_{bas}(C^\infty(T^2, \widehat Z)^{\widehat H}, \cL_{\widehat H, \tau})^{K_1+\tau K_2}[[u, u^{-1}]], -Q_{\widehat H, \tau}\right) 
$$
are quasi-isomorphic after localization. This {proves T-duality for $2d$  sigma models on circle bundles with $H$-flux}.

In \cite{HM21}, inspired by the theory of elliptic genus, we considered a graded version of T-duality of \cite{BEM04a}.  Namely we assembled various levels of the T-duality with flux $mH, m\in \ZZ$, constructed the {\bf graded Hori map} 
and showed that the compositions of the two way graded Hori maps are equal to the Euler operator, moreover the graded Hori maps preserve the Jacobi property of 
graded differential forms.

Motivated by the elliptic sheaves introduced by Berwick-Evans in \cite{DB19}, 
our results in \cite{HM21} can be summarized and recapped concisely as follows.  For the pair $(Z, H)$, define a sheaf $(\GG(Z, H), D^H)$ on $\HH$ of commutative differential graded algebras that to $U\subset \HH$ assigns the graded complex of $\cO(U)$-modules
\be   (\GG(Z, H)(U), D^H):= \bigoplus_{m\in \ZZ}\left(\cO(U; \Omega^{\bullet, \TT}(Z)[[u, u^{-1}]])\cdot y^m, \, d+u^{-1}mH \right),   \ee
where $y$ is a variable to indicate the grading.  
Dually, we can define the sheaf $(\GG(\widehat Z, \widehat H), D^{\widehat H}).$ Passing to cohomology, we get the sheave $\GGG(Z, H)$, or more precisely, define $\GGG(Z, H)$ by
\be \GGG(Z, H)(U):= \bigoplus_{m\in \ZZ}\cO(U; H(\Omega^{\bullet, \TT}(Z)[[u, u^{-1}]], d+u^{-1}mH))\cdot y^m.\ee
Similarly one has the sheaf $\GGG(\widehat Z, \widehat H)$ on the dual side. In \cite{HM21}, we showed that the {\bf graded twisted Chern character} of the Witten gerbe modules arising from a gerbe module pair $(E, E')$ give global sections of the sheaf  $(\GG(Z, H), D^H)$. Such global sections are the expansions at $y=0$ of Jacobi forms of index 0 of the two variables $(\tau, z)\in \HH\times \CC$ (with $y=e^{2\pi \sqrt{-1} z}$) over certain lattices when the degree 4 component of the twisted Chern character $Ch_H^{[4]}(E, E')$ is vanishing and $u=1$ (see Section \ref{graded} for details). In this framework, the graded Hori maps introduced in \cite{HM21} become the {\bf graded Hori morphisms} between the sheaves 
\be GHor_*: (\GG(Z, H), D^H)\to  (\GG(\widehat Z, \widehat H), D^{\widehat H}), \ GHor: \GGG(Z, H) \to \GGG(\widehat Z, \widehat H). \ee
\be \widehat{GHor}_*:  (\GG(\widehat Z, \widehat H), D^{\widehat H}) \to  (\GG(Z, H), D^H), \ \widehat{GHor}: \GGG(\widehat Z, \widehat H) \to \GGG(Z, H).\ee
The {\bf graded T-duality} theorem in \cite{HM21} states that 
\be \widehat{GHor}_* \circ GHor_*=-y\frac{\partial }{\partial y}, \ \  GHor_*\circ \widehat{GHor}_* =-y\frac{\partial }{\partial y},\ee
\be \label{gT} \widehat{GHor} \circ GHor=-y\frac{\partial }{\partial y}, \ \  GHor\circ \widehat{GHor} =-y\frac{\partial }{\partial y},\ee
and the graded Hori morphisms send a global section with Jacobi property to a global section with Jacobi property on the dual side. Restricted to $m=1$, on recovers the T-duality in (\ref{tdual}). 

In view of the Borel-Witten type localization (\ref{BWloc}), in this paper, we introduce sheaves using double loop spaces. Namely for the pair $(Z, H)$, define a sheaf $(\GG(\llz^H, \cL_{H}), \mathcal{Q}_H)$ on $\HH$ of commutative differential graded algebras that to $U\subset \HH$ assigns the graded complex of $\cO(U)$-modules
\be
(\GG(C^\infty(T^2, Z)^H, \cL_{H})(U), \mathcal{Q}_H):= \bigoplus_{m\in \ZZ}\left(\cO(U; \Omega^{\bullet, \TT}_{bas}(C^\infty(T^2, Z)^H, \cL_{H, \tau}^{\otimes m})[[u, u^{-1}]]^{K_1+\tau K_2})\cdot y^m,\  Q_{mH,\tau}\right),\\
\ee
where $\cO(U; \Omega^{\bullet, \TT}_{bas}(C^\infty(T^2, Z)^H, \cL_{H, \tau}^{\otimes m})[[u, u^{-1}]]^{K_1+\tau K_2})$ means 
i.e for each $\tau\in U$, one assigns to it an element in $\Omega^{\bullet, \TT}_{bas}(C^\infty(T^2, Z)^H, \cL_{H, \tau}^{\otimes m})[[u, u^{-1}]]^{K_1+\tau K_2}$. 
Dually, one can also define the sheaf $ (\GG(\llhatz^{\widehat H}, \cL_{\widehat H}), \mathcal{Q}_{\widehat H})$. Passing to cohomology, we get the sheaves
$\GGG(\llz^H, \cL_H)$ and $\GGG(\llhatz^{\widehat H}, \cL_{\widehat H})$.  The localisation theorem tells us that the restriction maps 
\be res: \GGG(\llz^H, \cL_H)\to \GGG(Z, H), \ \ \widehat{res}:  \GGG(\llhatz^{\widehat H}, \cL_{\widehat H})\to \GGG(\widehat Z, \widehat H)\ee
are isomorphisms of sheaves. Therefore we define the {\bf graded Hori morphisms for $2d$ $\sigma$-models} by
\be 
\begin{split}
&GHor^\sigma:= \widehat{res}^{-1}\circ GHor \circ res: \GGG(\llz^H, \cL_H)\to\GGG(\llhatz^{\widehat H}, \cL_{\widehat H}), \\
&\widehat{GHor}^\sigma:=res^{-1}\circ \widehat{GHor}\circ \widehat{res}: \GGG(\llhatz^{\widehat H}, \cL_{\widehat H}) \to \GGG(\llz^H, \cL_H),
\end{split}
\ee 
and assemble them into the following commutative diagram,
\[ 
\begin{tikzcd}
\GGG(\llz^H, \cL_H) \ar[rr, shift left]{rr}{GHor^\sigma} \ar[dd]{dd}{res\, \cong } & &\GGG(\llhatz^{\widehat H}, \cL_{\widehat H})\ar[ll, shift left]{ll}{ \widehat{GHor}^\sigma} \ar[dd]{dd}{\widehat{res}\, \cong }
\\
& &\\
\GGG(Z, H)\ar[rr, shift left]{rr}{GHor} & & \GGG(\widehat Z, \widehat H)
\ar[ll, shift left]{ll}{ \widehat{GHor}}
\end{tikzcd}
\] 

Combining (\ref{gT}), we obtain the following,
\begin{theorem}[\protect Theorem \ref{main},  graded T-duality with $H$-flux for $2d$ $\sigma$-models] One has
\be \widehat{GHor}^\sigma \circ GHor^\sigma=-y\frac{\partial }{\partial y}, \ \  GHor^\sigma\circ \widehat{GHor}^\sigma =-y\frac{\partial }{\partial y}.\ee
\end{theorem}

We would like to remark that when the torus degenerate to a circle, say, when the second circle in $T^2=S^1\times S^1$ degenerate to a point and consequently $K_2=0$, the theory in (\ref{BWloc}) degenerates to the holonomy line bundles and  completed periodic exotic twisted equivariant cohomology for (single) loop spaces developed in our earlier paper  \cite{HM15}, see also \cite{HM18} and \cite{JP}. The theory there is for $1d$ $\sigma$-model, while the theory in this paper is for $2d$ $\sigma$-model. 

Fix a spinor bundle $\mathbb{S}$ on $T^2$ and a Riemannian metric on $M$. The space of fields for the $\mathcal{N}=(0,1)$ sigma model with source $T^2$ and target $M$, denoted $\Phi_{T^2, M}$ consists of
\be (x, \psi)\in \Phi_{T^2, M}, \ \ x: T^2\to M, \ \  \psi\in \Gamma(T^2, \bar{\mathbb{S}}\otimes x^*TM),\ee
where $x$ is a smooth map and $\psi$ is a (anti-chiral) spinor valued in the pullback tangent
bundle. Let 
\be ev: C^\infty(T^2, M)\times T^2\to M\ee
be the evaluation map. Let $p_1: C^\infty(T^2, M)\times T^2\to C^\infty(T^2, M)$ and $p_2: C^\infty(T^2, M)\times T^2\to T^2$ be the obvious projection maps. Then the space of fields
$$ \Phi_{T^2, M}=\Gamma\left(C^\infty(T^2, M)\times T^2, p_2^* \bar{\mathbb{S}}\otimes ev^*TM \right).$$
The classical action is the function on fields
\be \mathcal{S}_{T^2}(x, \psi)=\frac{1}{2}\int_{T^2}\left( \langle \partial x, \bar\partial x\rangle+\langle \psi, \partial_\nabla \psi\rangle \right), \ \ \mathcal{S}_{T^2}\in C^\infty(\Phi_{T^2, M}), \ee
where $\partial_\nabla$ is the $\partial$-operator on $T^2$ twisted by the pullback of the Levi-Civita connection $\nabla$ on $TM$. Quantizing this classical field theory in the path integral formalism gives the famous Witten genus (c.f. \cite{DB19}). 

Let $p: C^\infty(T^2, M)^H\to C^\infty(T^2, M)$ be the circle bundle naturally arising form the flux $H$ constructed in Section \ref{construction}. Then we have $p\times id: C^\infty(T^2, M)^H\times T^2\to C^\infty(T^2, M)\times T^2$. Still denote by $p_1: C^\infty(T^2, M)^H\times T^2\to C^\infty(T^2, M)^H$ and $p_2: C^\infty(T^2, M)^H\times T^2\to T^2$ the obvious projection maps.  A natural interesting question is if we consider the fields, that are global sections in
$$ \Phi_{T^2, M, H, \tau}=\Gamma\left(C^\infty(T^2, M)^H\times T^2, p_1^*\cL_{H, \tau}\otimes p_2^*\bar{\mathbb{S}}\otimes (ev\circ (p\times id))^*TM \right), $$
what is the expression for the action functional for the corresponding sigma model? 
This would be useful for Duistermaat-Heckman type localization (cf. \cite{DH82}), which is not considered here.

We conjecture that with respect to the $\TT$-action on $Z$ and $\widehat \TT$-action on $\widehat Z$ respectively, that
$$
\left(\Omega^{\bullet, \TT}_{bas}(C^\infty(T^2, Z)^H, \cL_{H, \tau})^{K_1+\tau K_2}, Q_{H, \tau}\right) 
$$
and 
$$
\left(\Omega^{\bullet+1, \widehat \TT}_{bas}(C^\infty(T^2, \widehat Z)^{\widehat H}, \cL_{\widehat H, \tau})^{K_1+\tau K_2}, -Q_{\widehat H, \tau}\right) 
$$
are quasi-isomorphic without using the localization. This would prove a stronger version of T-duality for $2d$ sigma models on circle bundles with $H$-flux. 
The difficulty in proving this is that the putative Hori map in this context involves ``integration on double loop space". 
Results in this paper give evidence for this conjecture.

We mention that the sequence of ideas leading up to this paper is motivated in part by Witten \cite{W82}, Atiyah \cite{A85} and Bismut \cite{B85}.  Atiyah, working out an idea of Witten, revealed the remarkable fact that the index of the Dirac operator on the spin complex of a spin manifold can be formally interpreted as an integral of an equivariantly closed (with respect to the standard circle action on the loop space) differential form over loop space. A formal application of the localisation formula of Duistermaat-Heckman \cite{DH82} leads to the index theorem of Atiyah-Singer for the Dirac operator. Bismut extended this approach to a Dirac operator twisted by a vector bundle with connection. In doing so, for a vector bundle with connection, he constructed an equivariantly closed form on the loop space, lifting the Chern character form of the vector bundle with connection to the loop space.

\bigskip

\noindent{\em Acknowledgements.} Fei Han was partially supported by the grant AcRF R-146-000-218-112 from National University of Singapore. Varghese Mathai was supported by funding from the Australian Research Council, through the Australian Laureate Fellowship FL170100020. He thanks Kentaro Hori for suggesting this problem in a personal communication.


\section{Average $\tau$-holonomy line bundles and exotic theory on double loop spaces}\label{average}
\subsection{Some basics about double loop spaces}
Let $M$ be a smooth manifold and $T^2=S^1\times S^1$ the 2-dimensional torus. For simplicity, denote by $LLM=C^\infty(T^2, M)$, the double loop space of $M$, namely the space of smooth maps from $T^2$ to $M$. Let $K_1, K_2$ be the vector fields on $LLM$ obtained by rotating the first circle and second circle respectively. Clearly $[K_1, K_2]=0$.

Let $\{\sU_\alpha\}$ be an open cover of $M$. When the open cover $\{\sU_\alpha\}$ has some nice property,  $\{LL\sU_\alpha\}$ can be an open cover of $LLM$. For instance, if $\{\sU_\alpha\}$ is a maximal open cover of $M$ with the property that $H^i(U_{\alpha_I})=0$ for $i\geq 3$ where $U_{\alpha_I} = \bigcap_{i\in I} U_{\alpha_i}, \,$ $|I|<\infty$, then $\{LL\sU_{\alpha I}\}$ is an open cover of $LLM$. In fact, let $x:T^2\to M$ be a smooth loop in $M$ and $\sU_x$ a tubular neighbourhood of $x$ in $M$. Then $\{LL\sU_x, x\in LLM\}$ covers $LLM$. We call such a cover $\{\sU_\alpha\}$ a {\bf double loop Bryinski cover} for $M$.

Let $ev$ is the evaluation map
\beq \label{ev}
ev: LLM \times T^2 \to M: (x, s, t)\mapsto x(s, t).
\eeq
There are several operations constructed from the evaluation map. 
First we have the {\bf double transgression} map
\beq\label{eqn:trans}
\mu_{1, 2}: \Omega^\bullet(\sU_{\alpha_I} ) \longrightarrow \Omega^{\bullet-2}(LL\sU_{\alpha_I} )
\eeq
defined by
\beq
\mu_{1, 2}(\xi_I) = \int_{T^2} ev^*(\xi_I), \qquad \xi_I \in \Omega^\bullet(\sU_{\alpha_I} ).
\eeq

Another operation that we have is the {\bf averaging after transgression} map 
\beq\label{eqn:trans}
\overline{\mu_{1}}^2: \Omega^\bullet(\sU_{\alpha_I} ) \longrightarrow \Omega^{\bullet-1}(LL\sU_{\alpha_I} )
\eeq
defined by
\beq
\overline{\mu_{1}}^2(\xi_I) = \int_{S^1}\left(\int_{S^1} ev^*(\xi_I)\right)dt, \qquad \xi_I \in \Omega^\bullet(\sU_{\alpha_I} ),
\eeq
i.e.  integrate $ev^*(\xi_I)$ along the first circle and then average along the second circle. Similarly, one has 
\beq\label{eqn:trans}
\overline{\mu_{2}}^1: \Omega^\bullet(\sU_{\alpha_I} ) \longrightarrow \Omega^{\bullet-1}(LL\sU_{\alpha_I} ).
\eeq

Let $\omega \in \Omega^i(M)$. One also has the {\bf double loop averaging} map
$$ \overline{\overline{\omega}}:=\int_{T^2}ev^*(\omega)ds\wedge dt \in \Omega^i(LLM).$$
Clearly $L_{K_i} \overline{\overline{\omega}}= 0, i=1, 2$. Moreover it is not hard to see that
$$d\overline{\overline{\omega}}=\overline{\overline{d\omega}}, \ \ \mu_{1, 2}(\omega) =  \iota_{K_2}\iota_{K_1}\overline{\overline{\omega}}.$$ 

In addition to evaluation map (\ref{ev}),  there are also two partial evaluation maps
\be ev_1: LLM\times S^1\to LM, \ \ (x, s)\mapsto x(s, *), \ \ \ ev_2: LLM\times S^1\to LM, \ \ (x, t)\mapsto x(*, t).\ee
Let \be \label{pi_i} \pi_i: LLM\to LM,\  i=1, 2 \ee 
be defined by $\pi_1=ev_2|_{t=0}$, i.e restriction to the first circle and $\pi_2=ev_1|_{s=0},$ i.e restriction to the second circle. 

\subsection{Construction of the average $\tau$-holonomy line bundles and the exotic $\tau$-cohomology theory}\label{construction}

Suppose $M$ carries a gerbe \cite{Bry} with connection $(H, B_\alpha, F_{\alpha\beta}, (L_{\alpha\beta}, \nabla^{L_{\alpha\beta}}))$, with $H\in \Omega^3(M), B_\alpha\in \Omega^2(\sU_{\alpha})$ and $(L_{\alpha\beta}, \nabla^{L_{\alpha\beta}})$ being a complex line bundle over $\sU_{\alpha\beta}= \sU_\alpha\cap \sU_\beta$ such that
\be \label{gerbe}
\begin{split}
&H=dB_\alpha \ \mathrm{on}\  \sU_\alpha,\\
&B_\beta-B_\alpha=F_{\alpha\beta}=\left(\nabla^{L_{\alpha\beta}}\right)^2\ \ \mathrm{on}\  \sU_\alpha\cap \sU_\beta,\\
&(L_{\alpha\beta}, \nabla^{L_{\alpha\beta}})\otimes (L_{\beta\gamma}, \nabla^{L_{\beta\gamma}})\otimes (L_{\gamma\alpha}, \nabla^{L_{\gamma\alpha}})\simeq (\CC, d) \ \mathrm{on}\  \sU_\alpha\cap \sU_\beta\cap \sU_\gamma.
\end{split}
\ee

For any double loop $x\in LL\sU_\alpha\cap LL\sU_\beta$, i.e. $x: T^2\to \sU_{\alpha\beta}$, denote the holonomy of the $\nabla^{L_{\alpha\beta}}$ along the $K_1$-direction of by $hol^{1}$, which is a function of $t$; and the 
holonomy of the $\nabla^{L_{\alpha\beta}}$ along the $K_2$-direction by $hol^2$, which is a function of $s$. 
Consider the function on $LL\sU_\alpha\cap LL\sU_\beta$
\be \label{pre} g_{\alpha\beta}:=e^{\overline{\ln hol^1_{\alpha_\beta}}^2}\cdot e^{\tau\overline{\ln hol^2}^1_{\alpha_\beta}}.\ee
Note here for $\ln hol^1$, it is continuously defined for $t\in [0,\, 1)$, and for $\ln hol^2$, it is continuously defined for $s\in [0,\, 1)$. It is not hard to see that 
\be\label{inv}
\begin{split}
 &L_{K_1+\tau K_2}g_{\alpha\beta}\\
 =&L_{K_1+\tau K_2} \left(e^{\overline{\ln hol^1_{\alpha_\beta}}^2}\cdot e^{\tau\overline{\ln hol^2}^1_{\alpha_\beta}}\right)\\
 =&e^{\overline{\ln hol^1_{\alpha_\beta}}^2}\cdot e^{\tau\overline{\ln hol^2}^1_{\alpha_\beta}} [L_{\tau K_2}\overline{\ln hol^1_{\alpha_\beta}}^2+L_{K_1}\tau\overline{\ln hol^2}^1_{\alpha_\beta}]\\
 =&e^{\overline{\ln hol^1_{\alpha_\beta}}^2}\cdot e^{\tau\overline{\ln hol^2}^1_{\alpha_\beta}}[\tau\overline{K_2\ln hol^1_{\alpha_\beta}}^2+\tau\overline{K_1\ln hol^2}^1_{\alpha_\beta}]\\
 =&e^{\overline{\ln hol^1_{\alpha_\beta}}^2}\cdot e^{\tau\overline{\ln hol^2}^1_{\alpha_\beta}}\tau [\overline{\iota_{K_2}d\ln hol^1_{\alpha_\beta}}^2+
 \overline{\iota_{K_1}d\ln hol^2}^1_{\alpha_\beta}]\\
 =&e^{\overline{\ln hol^1_{\alpha_\beta}}^2}\cdot e^{\tau\overline{\ln hol^2}^1_{\alpha_\beta}}\tau \left[\overline{\iota_{K_2}\iota_{K_1}\overline{F_{\alpha\beta}}^1}^2+\overline{\iota_{K_1}\iota_{K_2}\overline{F_{\alpha\beta}}^2}^1\right]\\
 =&e^{\overline{\ln hol^1_{\alpha_\beta}}^2}\cdot e^{\tau\overline{\ln hol^2}^1_{\alpha_\beta}}\tau \left[ \iota_{K_1}\iota_{K_2}\overline{\overline{F_{\alpha\beta}}}+\iota_{K_2}\iota_{K_1}\overline{\overline{F_{\alpha\beta}}}\right]\\
 =&0,
 \end{split}
 \ee
where we have used the Brylinski's result (Proposition 6.2.2 in \cite{Bry}) on the differential of logarithm of holonomy function on loop space of a line bundle with connection. So $g_{\alpha\beta}$ is $(K_1+\tau K_2)$-invariant.

Denote
\be h_{\alpha\beta\gamma}:=g_{\alpha\beta}g_{\beta\gamma}g_{\gamma\alpha}. \ee
From (\ref{gerbe}), it is not hard to see that on the triple intersection $LL\sU_\alpha\cap LL\sU_\beta\cap LL\sU_\gamma$, $hol^i_{\alpha\beta}hol^i_{\beta\gamma}hol^i_{\gamma\alpha}=1, \ i=1, 2.$
Hence $\ln hol^i_{\alpha\beta}+\ln hol^i_{\beta\gamma}+\ln hol^i_{\gamma\alpha}\in 2\pi i\ZZ.$ As $\ln hol$'s are continuously defined, one must have
\be h_{\alpha\beta\gamma}=e^{2\pi im_{\alpha\beta\gamma}\tau }\ee
for some $m_{\alpha\beta\gamma}\in \ZZ$, where $\{m_{\alpha\beta\gamma}\}$ forms the $\mathrm{\check{C}}$ech cocycle representing $\pi_2^*(c_1(\cL_H))$ in $H^2(LLM, \ZZ)$ with $c_1(\cL_H)$ being the first Chern class of the holonomy line bundle $\cL_H$ on $LM$ arising from the gerbe data (\ref{gerbe}) on $M$ (\cite{Bry}, c.f. \cite{HM15} for details).

As $c_1(\cL_H)$ is nonzero, one cannot hope $h_{\alpha\beta\gamma}=1$, or in other words, $\{g_{\alpha\beta}\}$ do not form the gluing data of a line bundle. Nevertheless, if we pull back the line bundle $\cL_H$ to the circle bundle of itself, it becomes trivial and consequently the pull back of the first Chern class is 0. This motivates us to work on the circle bundle. Let $p:\cS_H\to LM$ be the circle bundle of the line bundle $\cL_H\to LM$. Then $p^*\cL_H$ is a trivial line bundle over $\cS_H$. Hence the class $p^*(c_1(\cL_H))$ be 0 on $\cS_H.$  Therefore $\widetilde{\pi}_2^*\circ p^*(c_1(\cL_H))$ is 0 on $p^*\cS_H$, the pulled back circle bundle over $LLM$. 
\begin{equation}\label{circlebundle}
\xymatrix @=4pc 
{
p^*\cS_H  \ar[r]^{\widetilde{\pi}_2} \ar[d]_{\widetilde{p}}  & \cS_H \ar[d]^{p}   \\ 
LLM \ar[r]_{\pi_2} & LM\\
}
\end{equation}
For simplicity, in the sequel we will denote the total space $p^*\cS_H$ by $LLM^H$, which carries the induced $T^2$-action arising from $LLM$. By abusing notations, still denote the two Killing vector fields on $LLM^H$ by $K_1, K_2$. As $\{m_{\alpha\beta\gamma}\}$ forms the $\mathrm{\check{C}}$ech cocycle representing $\pi_2^*(c_1(\cL_H))$,  the pulled back cochain $\widetilde{p}^*\{m_{\alpha\beta\gamma}\}$ must be exact on $LLM^H$. Let $\widetilde{p}^{-1}(LL\sU_{\alpha})$ be covered by the two open sets $\widetilde{p}^{-1}(LL\sU_{\alpha})^+, \widetilde{p}^{-1}(LL\sU_{\alpha})^-$ by removing the two opposite points on the fiber circles (with reference to the standard local Brylisnki basis of $\cS_H$). Then $\mathcal{U}=\{\widetilde{p}^{-1}(LL\sU_{\alpha})^+, \widetilde{p}^{-1}(LL\sU_{\alpha})^-\}$ form a good cover of $LLM^H$, which is a refinement of the cover $\{\widetilde{p}^{-1}(LL\sU_{\alpha})\}$. As $\widetilde{p}^*\{m_{\alpha\beta\gamma}\}$ is exact, there is a $\mathrm{\check{C}}$ech 2-cochain $f$ such that $(\delta f)(V_\alpha\cap V_\beta\cap V_\gamma)=m_{\alpha\beta\gamma}$, where $V_\alpha=\widetilde{p}^{-1}(LL\sU_{\alpha})^\pm$. 
Therefore the function $\widetilde{p}^*(g_{\alpha\beta})$ verifies
\be\label{lift} (\widetilde{p}^*(g_{\alpha\beta})e^{-2\pi i f(V_{\alpha\beta})\tau}) \cdot (\widetilde{p}^*(g_{\beta\gamma})e^{-2\pi i f(V_{\beta\gamma})\tau})\cdot  (\widetilde{p}^*(g_{\gamma\alpha})e^{-2\pi i f(V_{\gamma\alpha})\tau})=1. \ee

We construct a complex line bundle  $\cL_{H, \tau}$ on $LLM^H$, which we call {\bf average $\tau$-holonomy line bundle} as follows
\be \label{thelinebundle} \cL_{H, \tau}= \left(\coprod_{\alpha\in I} V_\alpha\times \CC\right)\bigg/\sim,  \ee
where
\be (\beta, x, w)\sim  (\alpha, x, (\widetilde{p}^*(g_{\alpha\beta})e^{-2\pi i f(V_{\alpha\beta})\tau}))(x)\cdot w), \ \ \forall \alpha, \beta\in I,\,  x\in V_{\alpha\beta}, w\in \CC. \ee
Denote by $s_\alpha=(x, 1)$ the local section of $\cL_{H, \tau}$ on $V_\alpha$.  Clearly $s_\alpha$'s are $(K_1+\tau K_2)$-invariant. 

Still using the Brylinski's result (Proposition 6.2.2 in \cite{Bry}) on the differential of logarithm of holonomy function on loop space of a line bundle with connection, and $B_\beta-B_\alpha=F_{\alpha\beta}$ in (\ref{gerbe}), we have
\be 
\begin{split}
&d\ln [\widetilde{p}^*(g_{\alpha\beta})e^{-2\pi i f(V_{\alpha\beta})\tau}]\\
=&\iota_{K_1}\widetilde{p}^*\overline{\overline{F_{\alpha\beta}}}+\tau \iota_{K_2}\widetilde{p}^*\overline{\overline{F_{\alpha\beta}}} \\
=&\left(-\iota_{K_1+\tau K_2}\widetilde{p}^*\overline{\overline{B_\alpha}} \right)-\left(-\iota_{K_1+\tau K_2}\widetilde{p}^*\overline{\overline{B_\beta}} \right).
\end{split}
\ee
Therefore, we see that $\{-\iota_{K_1+\tau K_2}\widetilde{p}^*\overline{\overline{B_\alpha}}\}$, as connection 1-forms under the local basis $\{s_\alpha\}$, patch together to be a connection on $\cL_{H, \tau}$, which we denote by $\nabla^{\cL_{H, \tau}}$. This connection is $(K_1+\tau K_2)$-invariant.

Denote by $\Omega^{\bullet}_{bas}(LLM^H,\cL_{H, \tau})$ the space of basic differential forms on $LLM^H$ with values in the average $\tau$-holonomy line bundle $\cL_{H, \tau}$. Here basic form means that contracted with vertical tangent vectors gives 0.   Let $u$ be an indeterminate such that $\deg u=2$. Consider the odd operator 
\be  Q_{H, \tau}:=\nabla^{\cL_{H, \tau}}-u\iota_{K_1+\tau K_2}+u^{-1}\widetilde{p}^*\overline{\overline{H}}\ee
which acts on $\Omega^{\bullet}(LLM^H,\cL_{H, \tau})_{bas}[[u, u^{-1}]].$
\begin{theorem} \label{flat} One has
\be Q_{H, \tau}^2+uL_{K_1+\tau K_2}^{\cL_{H, \tau}}=0.  \ee
\end{theorem}
\begin{proof} On $V_\alpha$, under the basis $s_\alpha$ (which is $(K_1+\tau K{_2})$-invariant), the operator 
$$\nabla^{\cL_{H, \tau}}-\iota_{K_1+\tau K_2}+\widetilde{p}^*\overline{\overline{H}}$$ can be written as 
$$d-\iota_{K_1+\tau K_2}\widetilde{p}^*\overline{\overline{B_\alpha}}-\iota_{K_1+\tau K_2}+\widetilde{p}^*\overline{\overline{H}}.$$
We have 
\be 
\begin{split}
&(d-\iota_{K_1+\tau K_2}\widetilde{p}^*\overline{\overline{B_\alpha}}-\iota_{K_1+\tau K_2}+\widetilde{p}^*\overline{\overline{H}})^2\\
=&-d\iota_{K_1+\tau K_2}\widetilde{p}^*\overline{\overline{B_\alpha}}-\iota_{K_1+\tau K_2}\widetilde{p}^*\overline{\overline{H}}-(d\iota_{K_1+\tau K_2}+\iota_{K_1+\tau K_2}d)\\
=&-d\iota_{K_1+\tau K_2}\widetilde{p}^*\overline{\overline{B_\alpha}}-\iota_{K_1+\tau K_2}\widetilde{p}^*\overline{\overline{dB_\alpha}}-L_{K_1+\tau K_2}\\
=&L_{K_1+\tau K_2}\widetilde{p}^*\overline{\overline{B_\alpha}}-L_{K_1+\tau K_2}\\
=&-L_{K_1+\tau K_2}.
\end{split}
\ee
The desired equality follows. 
\end{proof}

This theorem shows that $ Q_{H, \tau}$ is an {\em equivariantly flat
superconnection} (in the sense of Quillen \cite{Q, MQ}) on $\Omega^{\bullet}_{bas}(LLM^H,\cL_{H, \tau})^{K_1+\tau K_2}[[u, u^{-1}]]$. Therefore, we obtain a family of chain complex 
\be (\Omega^{\bullet}_{bas}(LLM^H,\cL_{H, \tau})^{K_1+\tau K_2}[[u, u^{-1}]], \ Q_{H, \tau})\ee
parametrized the $\tau\in \HH$. We call them {\em  completed periodic exotic twisted equivariant $\tau$-complex over double loop spaces} and their cohomology {\em completed periodic  exotic twisted equivariant $\tau$-cohomology over double loop spaces}. 

\subsection{Topology of the average $\tau$-holonomy line bundles} \label{explanation} In this Subsection, we give a topological understanding of the average $\tau$-holonomy line bundles. 

Still denote by $\cL_H$ the holonomy line bundle  on $LM$ arising from the gerbe data on $M$ (\cite{Bry}, c.f. \cite{HM15} for details) as in the previous subsection.  Our construction of the $\tau$-average holonomy line bundle $\cL_{H, \tau}$ in the above subsection is essentially a realization the following ``line bundle" over $LLM$
$$\pi_1^*(\cL_H)\otimes \pi_2^*(\cL_H)^{\otimes \tau},$$
i.e. $\pi_1^*(\cL_H)$ tensored with the ``$\tau$-th power" $\pi_2^*(\cL)^{\otimes \tau}$, where $\pi_i$'s are defined in (\ref{pi_i}). 

Since $\tau$ is not an integer but a complex number, $\pi_2^*(\cL_H)^{\otimes \tau}$ does not make sense. Nevertheless, under certain situation one can make sense out of this. Let $\xi$ be a complex line bundle over a manifold $X$.  Let $\mathfrak U=\{U_\alpha\}$ be a good open cover of $X$. Let $\{u_{\alpha\beta}\}$ be a system of $U(1)$-valued transition functions w.r.t $\mathfrak U$.  This gives us a closed Cech 1-cocycle $\{\theta_{\alpha\beta}\}$ valued in $\RR/\ZZ$ by taking $\theta_{\alpha\beta}=\frac{1}{2\pi i}\ln u_{\alpha\beta}$ in the argument range $[0, 2\pi)$ (c.f. Section 6 in \cite{BT}). So $\{\theta_{\alpha\beta}\}\in C^1(\mathfrak U, \RR/\ZZ)$. 

Let $0\to \ZZ\to \RR\to \RR/\ZZ\to 0$ be the obvious exact sequence. It is not hard to show that \newline
(1) $\{\theta_{\alpha\beta}\}$ can lifted as the image of a Cech cocyle in $C^1(\mathfrak U, \RR)$ $\iff$ $\xi$ is trivial; \newline
(2) different liftings differ by $\delta c$ for some $c\in C^1(\mathfrak U, \ZZ)$. 

Let $\xi$ be trivial and $\eta_{\alpha\beta} \in C^1(\mathfrak U, \RR)$ be a lifting of $\{\theta_{\alpha\beta}\}$. Then we can consider the $\CC$-valued functions $\{e^{2\pi i \tau \eta_{\alpha\beta}}\}$, which satisfies  
$$e^{2\pi i \tau \eta_{\alpha\beta}}\cdot e^{2\pi i \tau \eta_{\beta\gamma}}\cdot e^{2\pi i \tau \eta_{\gamma\alpha}}=1$$ 
and therefore patch to be a complex line bundle over $X$. One can consider it as the $\tau$-th power of $\xi$. When $\tau=\frac{1}{n}$, $n\in \ZZ$, the construction is just how we constructed the $n$-th roots of $\xi$. 

$\, $

Now come back to the double loop space. Although $\pi_2^*(\cL_H)^{\otimes \tau}$ does not make sense, but since (see (\ref{circlebundle})
$$\tilde{p}: p^*\mathcal{S}_H \to LLM$$ 
is the pull back of the circle bundle $\mathcal{S}_H$ on $LM$, $\tilde p^*(\pi_2^*(\cL_H))$ is a trivial bundle over $p^*\mathcal{S}_H$. Then 
$\tilde p^*(\pi_2^*(\cL_H))^{\otimes \tau}$ makes sense as explained above. So on $LLM^H:=p^*\mathcal{S}_H$, we have the line bundle 
$$ \tilde p^*(\pi_1^*(\cL_H))\otimes \tilde p^*(\pi_2^*(\cL_H))^{\otimes \tau}.$$ We claim that this line bundle is actually isomorphic to the average $\tau$-holonomy line bundle $\cL_{H, \tau}$  constructed in (\ref{thelinebundle}). 

The ``pre-transition functions" on $LLM$ defined (\ref{pre})  arise from the average process: 
$$g_{\alpha\beta}:=e^{\overline{\ln hol^1_{\alpha_\beta}}^2}\cdot e^{\tau\overline{\ln hol^2}^1_{\alpha_\beta}}.$$
The transition functions on $LLM^H$ defined in (\ref{lift}) are actually the transition functions for the line bundle 
$$ \tilde p^*(\pi_1^*(\cL_H))\otimes \tilde p^*(\pi_2^*(\cL_H))^{\otimes \tau},$$ because we do the average by continuous extension of $\ln hol^i_{\alpha_\beta}$ at time $0$, for  $i=1, 2$. Therefore the constructions of the $\tau$-holonomy line bundle $\cL_{H, \tau}$ in (\ref{thelinebundle}) is just a geometric realization of the line bundle $ \tilde p^*(\pi_1^*(\cL_H))\otimes \tilde p^*(\pi_2^*(\cL_H))^{\otimes \tau}.$ The purpose of the averaging process is to make all the data $(K_1+\tau K_2)$-invariant as shown in (\ref{inv}).

\section{Localisation} \label{localisation}
The $T^2$-action on $LLM^H$ has the fixed point set $M\times S^1$. On the fixed point set, the map $\widetilde{p}$ becomes the projection $p:M\times S^1\to M$, and the complex 
$$ (\Omega^{\bullet}_{bas}(LLM^H,\cL_{H, \tau})^{K_1+\tau K_2}[[u, u^{-1}]], \ Q_{H, \tau})$$ 
becomes the complex 
$$\left(\Omega^{\bullet}_{bas}( M\times S^1)[[u, u^{-1}]], d+u^{-1}p^*H\right).$$ 
The purpose of this section is to establish the following Borel-Witten type localisation theorem. 
\begin{theorem}\label{localization} Let $i: M\times S^1\to LLM^H$ be the inclusion map. Then the restriction map
\be 
\begin{split}
i^*: \left(\Omega^{\bullet}_{bas}(LLM^H, \cL_{H, \tau})^{K_1+\tau K_2}[[u, u^{-1}]], Q_{H, \tau}\right)\to & \left(\Omega^{\bullet}_{bas}( M\times S^1)[[u, u^{-1}]], d+u^{-1}p^*H\right)\\
& =\left(\Omega^{\bullet}( M)[[u, u^{-1}]], d+u^{-1}H\right)
\end{split}
\ee
is a quasi-isomorphism, $\forall \tau\in \HH$. 
\end{theorem}

We will prove a theorem under the setting as follows. Let $Y$ be a (possibly infinite dimensional) $T^2$-manifold and $\xi$ a $T^2$-equivariant complex line bundle over $Y$. Let $\nabla^\xi$  be a $(K_1+\tau K_2)$-invariant connection on $\xi$  and $\chi\in \Omega^3_{cl}(Y)$ a closed $T^2$-invariant 3-form such that 
$$(\nabla^{\xi}-u\iota_{K_1+\tau K_2}+u^{-1}\chi)^2+uL_{K_1+\tau K_2}^\xi=0,$$ 
where $K_1, K_2$ are the two Killing vector fields resulted from the actions of the two factor circles in $T^2$. Denote by

\be h_{\tau}^\bu(Y, \nabla^\xi:\chi):=H^\bu(\Omega^\bullet(Y, \xi)^{K_1+\tau K_2}[[u, u^{-1}]], \nabla^{\xi}-u\iota_{K_1+\tau K_2}+u^{-1}\chi), \ee
the cohomology of the complex.

Further assume that $Y$ is {\em strongly $T^2$-regular},  i.e. the $T^2$-fixed point set $F$ is a smooth submanifold, which has a $T^2$-invariant neighbourhood $N$ such that \newline (i) $N$ has an invariant good cover $\{W_\alpha\}$, i.e. each $W_\alpha$ is $T^2$-homotopic to a point; 
\newline (ii) there exists a projection $p:N\to F$ and an equivariant homotopy
$$g: N\times I\to N, $$
($I=[0,1]$) with the properties that
$$g_0=i\circ p,\quad g_1=id, $$
$$p(g_t(x))=g_0(x), \quad\forall x\in N, t\in I,$$ where $i:F\to N$ is the embedding. \newline

\begin{theorem}\label{Y-localisation} If $Y$ is a strong regular $T^2$-manifold, then
 \be i^*: h^\bu_{\tau}(Y, \nabla^\xi:\chi) \cong H(\Omega^\bu(F, i^*\xi)[[u, u^{-1}]], i^*\nabla^\xi+u^{-1}i^*\chi).\ee
 is an isomorphism, where for simplicity, we also denote the embedding of $F$ into $Y$ by $i$.
\end{theorem}

Note that the double loop space $LLM$ as well as $LLM^H$ are such strongly $T^2$-regular manifolds. To prove Theorem \ref{localization}, we will just need to take $Y=LLM^H$ and use a ``basic" form version of Theorem \ref{Y-localisation}, which is not hard to show by slighltly modifying the following proof of Theorem \ref{Y-localisation}. 

Now we prove the localization theorem \ref{Y-localisation}. We first establish some properties of the theory in the following lemmas. 

\begin{lemma} Let $Y$ be a $T^2$-manifold and $\xi$ be an $T^2$-equivariant complex line bundle over $Y$. Let  $\nabla^\xi $ be an $(K_1+\tau K_2)$-invariant connection on $\xi$  and $\chi\in \Omega^3(Y)$ be a $T^2$-invariant $3$-form such that 
$$(\nabla^{\xi}-u\iota_{K_1+\tau K_2}+u^{-1}\chi)^2+uL_{K_1+\tau K_2}^\xi=0.$$
Suppose the $T^2$-action has no fixed points, then
$$h_{\tau}^\bu(Y, \nabla^\xi: \chi)=0. $$
\end{lemma}
\begin{proof}  Choose a $T^2$-invariant metric on $TY$ and extend it to be a Hermitian metric on $TY\otimes \CC$.  Let $\theta$ be the one form dual to the vector field $\ktau$ with respect to the Hermitian metric. Note that $\theta$ is $T^2$-invariant. Then we have
\be
\begin{split} &(d-u\iota_\ktau)\theta\\
=&d\theta-u\langle \ktau, K_1+\bar \tau K_2\rangle\\
=&d\theta-u(|K_1+aK_2|^2+b^2|K_2|^2)\\
=&-u(|K_1+aK_2|^2+b^2|K_2|^2)\left(1-\frac{d\theta}{u(|K_1+aK_2|^2+b^2|K_2|^2)},\right)
\end{split}
\ee
where  $\tau=a+b\sqrt{-1}$, with $a, b\in \RR, b>0$ (noting $b\neq 0$ implies that $|K_1+aK_2|^2+b^2|K_2|^2$ must be nonzero when $Y$ does not have $T^2$-fixed points).

Let 
$$\gamma=((d-u\iota_\ktau)\theta)^{-1}=-u^{-1}(|K_1+aK_2|^2+b^2|K_2|^2)^{-1}
\sum_{i=0}^\infty\left(\frac{d\theta}{u(|K_1+aK_2|^2+b^2|K_2|^2)}\right)^i.$$ As $\gamma((d-u\iota_\ktau)\theta)=1$, applying $d-u\iota_\ktau$ on both sides and using the fact that $L_\ktau \theta=0$, we have
$(d-u\iota_\ktau)\gamma=0$.

Define $\omega=\theta\gamma$, an odd degree form.  Then
$(d-u\iota_\ktau)\omega=((d-u\iota_\ktau))\gamma=1.$

 $\forall x\in \Omega^\bu(Y, \xi),$ we have 
 \be
 \begin{split}
 &(\nabla^{\xi}-u\iota_\ktau+u^{-1}\chi)(\omega\cdot x)+\omega\cdot( (\nabla^{\xi}-u\iota_\ktau+u^{-1}\chi)x)\\
 =& (\nabla^{\xi}-u\iota_\ktau)(\omega\cdot x)+(u^{-1}\chi\omega)\cdot x+\omega\cdot( (\nabla^{\xi}-u\iota_\ktau)x)+(u^{-1}\omega \chi)\cdot x\\
 =&(\nabla^{\xi}-u\iota_\ktau)(\omega\cdot x)+\omega\cdot( (\nabla^{\xi}-u\iota_\ktau)x)\\
 =&((d-u\iota_\ktau)\omega)\cdot x-\omega\cdot ((\nabla^{\xi}-u\iota_\ktau)x)+\omega\cdot( (\nabla^{\xi}-u\iota_\ktau)x)\\
 =&x.
 \end{split}
 \ee
This homotopy tells us that  $h_\tau^\bu(Y, \nabla^\xi: \chi)=0. $

\end{proof}

Next we will establish the Mayer-Vietoris property.  Suppose $Y=U\cup V$, with $U, V$ being open. Assume that $U, V$ and $U\cap V$ are all $T^2$-invariant submanifolds.  We have the following exact sequence,
\be 
\begin{split}\label{shortMV} 0\to\Omega^\bu(Y, \xi)^\ktau[[u, u^{-1}]]\stackrel{r}{\rightarrow}\Omega^\bu(U, \xi)^\ktau[[u, u^{-1}]]&\oplus\Omega^\bu(V,\xi)^\ktau[[u, u^{-1}]]\\ 
&\stackrel{e}{\rightarrow}\Omega^\bu(U\cap V, \xi)^\ktau[[u, u^{-1}]]\to0,
\end{split}
\ee
where $r(\eta)=(\eta|_U, \eta|_V)$, $e(\eta_1, \eta_2)=\eta_2-\eta_1$. 
The exactness can be seen from the following. Take any $f\in \Omega^\bu(U\cap V, \xi)^\ktau$, let $\{\rho_U, \rho_V\}$ be a partition of unity subordinate to the open cover $\{U, V\}$. We can assume that $\rho_U, \rho_V$ are both $T^2$-invariant, otherwise, simply average them over $T^2$.  Then $(-\rho_Vf,\rho_Uf)\in  \Omega^\bu(U, \xi)^\ktau\oplus\Omega^\bu(V,\xi)^\ktau$ maps onto $f$.

\begin{lemma} The short exact sequence (\ref{shortMV}) induces a long exact Mayer-Vietoris sequence in cohomology,
\be
 \xymatrix{
h_\tau^{ev}(Y,\nabla^\xi:\chi) \ar[r]& h_\tau^{ev}(U,\nabla^\xi:\chi) \oplus h_\tau^{ev}(V,\nabla^\xi:\chi) \ar[r]& h_\tau^{ev}(U\cap V,\nabla^\xi:\chi) \ar[d]^{D^*}\\
h_\tau^{odd}(U\cap V,\nabla^\xi:\chi)\ar[u]^{D^*} &  \ar[l] h_\tau^{odd}(U,\nabla^\xi:\chi) \oplus h_\tau^{odd}(V,\nabla^\xi:\chi) & \ar[l] h_\tau^{odd}(Y,\nabla^\xi:\chi)
}
\ee
\end{lemma}

\begin{proof} Let $\omega\in \Omega^{ev}(U\cap V, \xi)^\ktau[[u, u^{-1}]]$ such that $(\nabla^{\xi}-u\iota_\ktau+u^{-1}\chi)\omega=0$. Let $\{\rho_U, \rho_V\}$ be a partition of unity subordinate to the open cover $\{U, V\}$ such that $\rho_U, \rho_V$ are both $T^2$-invariant. Define the coboundary operator by
\be
D^*([\omega])=\left\{
\begin{array}{ccc}
[-(\nabla^{\xi}-u\iota_\ktau+u^{-1}\chi)(\rho_V\omega)]& \mathrm{on} \ U\\

[(\nabla^{\xi}-u\iota_\ktau+u^{-1}\chi)(\rho_U\omega)]& \mathrm{on} \ V
\end{array}
\right.
\ee
which is easily seen to be an element in $h_\tau^{odd}(Y,\nabla^\xi:\chi)$ and is independent of choices in this construction.  Similarly one can define the coboundary operator on $h_\tau^{odd}(U\cap V,\nabla^\xi:\chi)$.

It is not hard to check the exactness of the sequence.
\end{proof}

In the following two lemmas, we will establish homotopy properties of the theory. 

\begin{lemma}Let $Y$ be a $T^2$-manifold and $\xi$ be a $\TT$-equivariant complex line bundle over $Y$. Let  $\nabla^\xi $ be a $(K_1+\tau K_2)$-invariant connection on $\xi$  and $\chi\in \Omega^3(Y)$ be a $T^2$-invariant $3$-form such that $$(\nabla^{\xi}-u\iota_\ktau+u^{-1}\chi)^2+uL_\ktau^\xi=0.$$
Let $i_0:Y \rightarrow Y\times I,  m \mapsto (m,0)$ be the inclusion and $\pi: Y\times I\to Y$ be the projection. $T^2$ acts on $Y\times I$ in the obvious way. We have
$$h^\bu_\tau(Y\times I, \pi^*\nabla^{\xi}:\pi^*\chi)\cong h^\bu_\tau(Y, \nabla^\xi:\chi). $$

\end{lemma}

\begin{proof}It is clear that $i_0^*\circ \pi^*=id: \Omega^\bu(Y, \xi)\to \Omega^\bu(Y, \xi).$

We will show that $\pi^*\circ i_0^*$ is homotopic to identity on $\Omega^\bu(Y\times I, \pi^*\xi)$.

Choose an atlas $\{U_\alpha\}$ for $Y$, then $\{U_\alpha\times I\}$ is an atlas on $Y\times I$. Let $\{s_\alpha\}$ be local basis of $\xi$ on $\{U_\alpha\}$ and $\{g_{\alpha\beta}\}$ be the transition functions. Then $\{\pi^*s_\alpha\}$ are local basis for the bundle $\pi^*\xi$ on $\{U_\alpha\times I\}$ and $\{\pi^*g_{\alpha\beta}\}$ are the transition functions.

For any $\omega\in \Omega^\bu(Y\times I, \pi^*\xi)$, define $K_0\omega \in \Omega^\bu(Y\times I, \pi^*\xi)$ in the following way. Let $\omega=\omega_\alpha\otimes (\pi^*s_\alpha)$ on $U_\alpha\times I$. If $\omega_\alpha$ is of the form $(\pi^*\psi)f(x,t)$, set $K_0(\omega)|_{U_\alpha\times I}=0$; if $\omega_\alpha$ is of the form $(\pi^*\psi)f(x,t)dt$, set $P_0(\omega)|_{U_\alpha\times I}=((\pi^*\psi)\int_0^tf(x,t)dt)\otimes (\pi^*s_\alpha).$ As the transition function from $\pi^*s_\alpha$ to $\pi^*s_\beta$ is $\pi^*g_{\alpha\beta}$, it is not hard to see that $P_0(\omega)|_{U_\alpha\times I}$ patch together to give $P_0(\omega)\in \Omega^\bu(Y\times I, \pi^*\xi)$.

By (\cite{BT}, Sec 4), we know that on $U_\alpha\times I$, (for simplicity, we also denote by $P_0$ the similar operator on $\Omega^\bu(Y\times I)$),
\be (1-\pi^*\circ i_0^*)\omega_\alpha=(-1)^{p(\omega_\alpha)-1}(dP_0-P_0d)\omega_\alpha.  \ee
Since the connection $\pi^*\nabla^{\xi}$ on $\pi^*\xi$ is horizontal, we have
\be (1-\pi^*\circ i_0^*)\omega=(-1)^{p(\omega)-1}((\pi^*\nabla^{\xi})\circ P_0-P_0\circ(\pi^*\nabla^{\xi}))\omega.  \ee

Moreover, as $\iota_\ktau$ and $\pi^*\chi$ are both horizontal, they both commute with $P_0$ and therefore,
\be (1-\pi^*\circ i_0^*)\omega=(-1)^{p(\omega)-1}((\pi^*\nabla^{\xi}-u\iota_{\ktau}+u^{-1}\pi^*\chi)\circ P_0-P_0\circ(\pi^*\nabla^{\xi}-u\iota_\ktau+u^{-1}\pi^*\chi))\omega.  \ee

The isomorphism therefore follows.

\end{proof}

\begin{lemma}Let $Y$ be a $T^2$-manifold with an invariant good cover $\{W_\alpha\}$. $\TT$ acts on $Y\times I$ in the obvious way.
Let $\xi$ be an $T^2$-equivariant complex line bundle over $Y\times I$ equipped with a $(K_1+\tau K_2)$-invariant connection $\nabla^\xi $ and $\chi\in \Omega^3_{cl}(Y\times I)$ a closed $T^2$-invriant $3$-form on $Y\times I$ such that 
$$(\nabla^{\xi}-u\iota_\ktau+u^{-1}\chi)^2+uL_\ktau^\xi=0.$$
Let $i_0:Y \rightarrow Y\times I,  m \mapsto (m,0)$ and $i_1:Y \rightarrow Y\times I,  m \mapsto (m,1)$ be the inclusions. We have
$$h^\bu_\tau(Y, i_0^*\nabla^{\xi}:i_0^*\chi)\cong h^\bu_\tau(Y, i_1^*\nabla^{\xi}:i_1^*\chi). $$

\end{lemma}
\begin{proof} Let $\{s_\alpha\}$ be local $(K_1+\tau K_2)$-invariant basis of $i_0^*\xi$ on $\{W_\alpha\}$ and $\{g_{\alpha\beta}\}$ be the transition functions. Let $\Theta_\alpha$ be the basis of bundle $\xi$ on $W_\alpha\times I$ such that
\be
\left\{
\begin{array}{cc}
&\nabla^\xi_{\frac{\partial}{\partial t}}\Theta_\alpha=0\\
&\Theta_\alpha|_{Y\times \{0\}}=s_\alpha.
\end{array}
\right.
\ee
Let $\pi: Y\times I\to Y$ be the projection.  Since $\pi^*g_{\alpha\beta}$ are horizontal, we can see that $\{\pi^*g_{\alpha\beta}\}$ are transition functions of the local basis $\{\Theta_\alpha\}$. Let $\{\theta_\alpha\}$ be the connection one form of $\nabla^{\xi}$ with respect to the local basis $\{\Theta_\alpha\}.$ Obviously, they are horizontal forms. Also as $s_\alpha$ are invariant sections and $\nabla^\xi$ is an invariant connection, $\Theta_\alpha$ are also local invariant basis of $\xi$ with respect to the cover $\{W_\alpha\times I\}$.

Take any $\omega\in \Omega^\bu(Y, i_0^*\xi)$. Let $\omega=\omega_\alpha\otimes s_\alpha$ on $W_\alpha$. Define
$$Q(\omega)|_{U_\alpha\times I}=(\pi^*\omega_\alpha)\otimes \Theta_\alpha.$$ It is clear that $Q(\omega)|_{W_\alpha\times I}$ patch together to give $Q(\omega)\in \Omega^\bu(Y\times I, \xi).$

Define
\be \rho^0_1: \Omega^\bu(Y, i_0^*\xi)\to  \Omega^\bu(Y, i_1^*\xi) \ee
$$\omega\mapsto e^{-u^{-1}P_0\chi}Q(\omega)|_{Y\times \{1\}}, $$
where $P_0$ is defined as in the proof of Lemma 2.5.

Suppose $(i_0^*\nabla^\xi-u\iota_\ktau)\omega=-u^{-1}(i_0^*\chi)\cdot \omega.$ Locally this means that
\be (d+i_0^*\theta_\alpha-u\iota_\ktau)\omega_\alpha=-u^{-1}(i_0^*\chi)\cdot \omega_\alpha.  \ee
Then if we compute $(\nabla^\xi-u\iota_\ktau)(e^{-u^{-1}P_0\chi}Q(\omega) )$, locally under the basis $\Theta_\alpha$, we have
\be
\begin{split}
&(d+\theta_\alpha-u\iota_\ktau)(e^{-u^{-1}P_0\chi}\pi^*(\omega_\alpha))\\
=&e^{-u^{-1}P_0\chi}(-u^{-1}d(P_0\chi))\pi^*\omega_\alpha+e^{-u^{-1}P_0\chi}\pi^*(d\omega_\alpha)\\
&+(\theta_\alpha-\pi^*i_0^*\theta_\alpha)e^{-u^{-1}P_0\chi}\pi^*(\omega_\alpha)+e^{-u^{-1}P_0\chi}\pi^*(i_0^*\theta_\alpha\omega_\alpha)\\
&+e^{-u^{-1}P_0\chi}(\iota_\ktau P_0\chi)\pi^*(\omega_\alpha)-ue^{-u^{-1}P_0\chi}\pi^*(\iota_\ktau \omega_\alpha)\\
=&e^{-u^{-1}P_0\chi}[(-u^{-1}d(P_0\chi)-u^{-1}\pi^*\circ i_0^*\chi)+(\theta_\alpha-\pi^*\circ i_0^*\theta_\alpha+\iota_\ktau P_0\chi)]\pi^*\omega_\alpha.
\end{split}
\ee

However by homotopy formula, we have $ \chi-\pi^*\circ i_0^*\chi=dP_0\chi-P_0d\chi=dP_0\chi.$ So
$$dP_0\chi+\pi^*\circ i_0^*\chi=\chi.$$

Moreover, by homotopy formula, $\theta_\alpha-\pi^*\circ i_0^*\theta_\alpha=dP_0\theta_\alpha-P_0d\theta_\alpha. $ As $\theta_\alpha$ is horizontal, $P_0\theta_\alpha=0$. So $\theta_\alpha-\pi^*\circ i_0^*\theta_\alpha=-P_0d\theta_\alpha.$

But
$(\nabla^{\xi}-u\iota_\ktau+u^{-1}\chi)^2+uL_\ktau^\xi=0$ tells us that $d\theta_\alpha-\iota_\ktau \chi=0$. Noticing that $\ktau$ is horizontal,  so $\iota_\ktau P_0=P_0\iota_\ktau$. Therefore we have
\be \theta_\alpha-\pi^*\circ i_0^*\theta_\alpha+\iota_\ktau P_0\chi=-P_0d\theta_\alpha+P_0\iota_\ktau \chi=-P_0d\theta_\alpha+P_0d\theta_\alpha=0. \ee
So
\be (d+\theta_\alpha-u\iota_\ktau)(e^{-u^{-1}P_0\chi}\pi^*(\omega_\alpha))=-u^{-1}\chi(e^{-u^{-1}P_0\chi}\pi^*(\omega_\alpha)),\ee
which gives
\be  (\nabla^\xi-u\iota_\ktau)(e^{-u^{-1}P_0\chi}Q(\omega))=-u^{-1}\chi(e^{-u^{-1}P_0\chi}Q(\omega)). \ee

This shows that $\rho_1^0$ gives us a homomorphism
\be  \rho_1^0: h^\bu_\tau(Y, i_0^*\nabla^{\xi}:i_0^*\chi)\to h^\bu_\tau(Y, i_1^*\nabla^{\xi}:i_1^*\chi).  \ee

It is easy to see that we can define the inverse $\rho_0^1$ in a similar manner. Just need to replace the operator $P_0$ by $P_1$ with replacing $(\pi^*\psi)\int_0^tf(x,t)dt$ by $(\pi^*\psi)\int_t^0f(x,t)dt$ in the definition of $P_0.$
\end{proof}

\begin{remark}We would like to point out that there exist $T^2$-manifolds with invariant good covers. For instance, the $T^2$-invariant tubular neighborhoods of a finite dimensional manifold in its double loop space.  
\end{remark}

Now we are ready to prove Theorem \ref{Y-localisation}. 

Pick an invariant tubular neighbourhood of $N$ of $F$ such that there exists a projection $p:N\to F$ and a $T^2$-equivariant homotopy
$$g: N\times I\to N, $$ with the property that
$$g_0=i\circ p, g_1=id, $$
$$p(g_t(x))=g_0(x), \forall x\in N, t\in I. $$

As $i\circ p=id$, we have
$$i^*p^*=id: h^\bu_\tau(F, i^*\nabla^\xi, i^*\chi) \to h^\bu_\tau(N, p^*i^*\nabla^\xi,p^* i^*\chi)\to h^\bu_\tau(F, i^*\nabla^\xi, i^*\chi).$$

As $p(g_t(x))=g_0(x), \forall x\in N, t\in I, $ we have $i\circ p\circ g=i\circ p\circ \pi. $ So
\be g^*p^*i^*(\xi)=\pi^*(p^*i^*(\xi)),   g^*p^*i^*(\chi)= \pi^*(p^*i^*(\chi)).          \ee

By Lemma 2.5, we have
\be h^\bu_\tau(N, p^*i^*\nabla^\xi: p^*i^*\chi)\cong h^\bu_\tau(N\times I, g^*p^*i^*(\nabla^\xi): g^*p^*i^*(\chi)).\ee

Suppose $i_0:N \rightarrow N\times I,  m \mapsto (m,0)$ and $i_1:N \rightarrow N\times I,  m \mapsto (m,1)$ be the inclusions. Since $i\circ p=g\circ i_0$, we have
\be i_0^*g^*=p^*i^*: h^\bu_\tau(N, p^*i^*\nabla^\xi: p^*i^*\chi)\to h^\bu_\tau(F, i^*\nabla^\xi, i^*\chi) \to h^\bu_\tau(N, p^*i^*\nabla^\xi: p^*i^*\chi). \ee
By Lemma 2.5, $i_0^*$ and $i_1^*$ are both inverse to
$$\pi^*:h^\bu_\tau(N, p^*i^*\nabla^\xi: p^*i^*\chi)\to h^\bu_\tau(N\times I, g^*p^*i^*(\nabla^\xi): g^*p^*i^*(\chi)). $$ So
$$p^*i^*=i_1^*g^*=(g\circ i_1)^*=id:  h^\bu_\tau(N, p^*i^*\nabla^\xi: p^*i^*\chi)\to h^\bu_\tau(N, p^*i^*\nabla^\xi: p^*i^*\chi).$$

Therefore we have
\be  h^\bu_\tau(F, i^*\nabla^\xi: i^*\chi) \cong h^\bu_\tau(N, p^*i^*\nabla^\xi: p^* i^*\chi).   \ee

On the other hand, $g^*\xi, g^*(\nabla^\xi), g^*(\chi)$ have the property that
$$g^*\chi\in \Omega^3_{cl}(N\times I), \ (g^*\nabla^\xi-u\iota_\ktau+u^{-1}g^*\chi)^2+L_\ktau^{g^*\xi}=0.$$
Also $$i_0^*g^*\xi=p^*i^*\xi, i_0^*(g^*(\chi))=p^*i^*(\chi),$$
$$i_1^*g^*\xi=\xi, i_1^*(g^*(\chi))=\chi.$$
Then by Lemma 2.6, we have
\be  h^\bu_\tau(N, p^*i^*\nabla^\xi: p^* i^*\chi)\cong h^\bu_\tau(N, \nabla^\xi: \chi) .   \ee

So we can see that
\be h^\bu_\tau(F, i^*\nabla^\xi:  i^*\chi) \cong h^\bu_\tau(N, \nabla^\xi: \chi).\ee

Note that $Y=(Y\setminus F)\cup N$ and $Y\setminus F$, $(Y\setminus F)\cap N=N\setminus F$ are fixed point free. Then Lemma 2.3 and Lemma 2.4 (the Mayer-Vietoris sequence) as well as the above isomorphism give that
$$ i^*: h^\bu_\tau(Y, \nabla^\xi: \chi)\cong h^\bu_\tau(N, \nabla^\xi:\chi)\cong h^\bu_\tau(F, i^*\nabla^\xi: i^*\chi)$$
is an isomorphism.


\section{Graded T-duality for $2d$ $\sigma$-models} 
\subsection{Review of T-duality}\label{revT} Let $Z$ be a smooth manifold endowed with an $H$-flux which is a presentative in the degree 3 Deligne cohomology of $Z$, that is $H\in \Omega^3(Z)$ with integral periods (for simplicity, we drop the factor of $\frac{1}{2\pi \sqrt{-1}}$). 

Here we briefly review topological T-duality arising for the case of principal circle bundles with a $H$-flux. For such a case, one begins with a principal circle bundle $\pi: Z \to X$ whose first Chern class is given by $[F] \in H^2(X, \mathbb{Z})$, along with a $H$-flux given by some $[H] \in H^3(Z, \mathbb{Z})$. The aim is to then determine the corresponding data arising after an application of the T-duality transformation. 

One method for determining the T-dual data is by focusing on the Gysin sequence associated to the bundle $\pi : Z \to X$, given by:
\begin{center}
\begin{tikzcd}
\cdots \arrow[r] & H^{3}(X, \mathbb{Z}) \arrow[r, "\pi^{*}"] & H^3(Z, \mathbb{Z}) \arrow[r, "\pi_{*}"] & H^{2}(X, \mathbb{Z}) \arrow[r, "\lbrack F \rbrack \wedge"] & H^{4}(X, \mathbb{Z})  \arrow[r] & \cdots
\end{tikzcd}
\end{center}
Letting $[H] \in H^3(Z, \mathbb{Z})$, define $[\widehat{F}] = \pi_*([H]) \in H^2(X, \mathbb{Z})$, and make the choice of some principal circle bundle $\widehat{\pi} : \widehat{Z} \to X$ with first Chern class $[\widehat{F}]$. Having made such a choice, we then consider the Gysin sequence associated to the bundle $ \widehat{Z}$ over $ X$,
\begin{center}
\begin{tikzcd}
\cdots \arrow[r] & H^{3}(X, \mathbb{Z}) \arrow[r, "\widehat{\pi}^{*}"] & H^3(\widehat{Z}, \mathbb{Z}) \arrow[r, "\widehat{\pi}_{*}"] & H^{2}(X, \mathbb{Z}) \arrow[r, "\lbrack \widehat{F}\rbrack \wedge"] & H^{4}(X, \mathbb{Z})  \arrow[r] & \cdots
\end{tikzcd}
\end{center}
Now using exactness, and the fact that $[F] \wedge [\widehat{F}]= [F] \wedge \pi_*([H])=0$, there exists an $[\widehat{H}] \in H^3(\widehat{Z})$ such that $[F] = \widehat{\pi}_*([\widehat{H}])$. The following theorem gives a a global, geometric version of the Buscher rules \cite{Buscher}.

\begin{theorem}[\cite{BEM04a}]\label{TD1}
Let $\pi : Z \to X $ denote a principal circle bundle whose first Chern class is given by $[F] \in H^2(X, \mathbb{Z})$, and let $[H] \in H^3(Z)$ denote a $H$-flux on $Z$. 

Then there exists a T-dual bundle $\widehat{\pi} : \widehat{Z} \to X$ whose first Chern class is denoted $[\widehat{F}] \in H^2(X, \mathbb{Z})$ and a T-dual $H$-flux on this bundle given by $[\widehat{H}] \in H^3(\widehat{Z}, \mathbb{Z})$, satisfying
\begin{align*}
[\widehat{F}] &= \pi_*([H]),\\
[F] &= \widehat{\pi}_*([\widehat{H}]).
\end{align*}

Furthermore, letting $Z \times_X \widehat{Z} = \{(a, b) \in Z \times \widehat{Z} | \pi(a) = \widehat{\pi}(b)\}$ and considering the following commutative diagram of bundle maps 
\begin{center}
\begin{tikzcd}
 &Z \times_X \widehat{Z} \arrow[dl, "p"'] \arrow[rd, "\widehat{p}"] & \\
Z\arrow[rd, "\pi"] & & \widehat{Z}, \arrow[ld, "\widehat{\pi}"'] \\
 & X & 
\end{tikzcd}
\end{center}
then if the two $H$-fluxes, $[H]$ and $[\widehat{H}]$, satisfy
\begin{align*}
p^*([H]) = \widehat{p}^*([\widehat{H}]),
\end{align*}
the T-dual pair is unique up to bundle automorphism, and thus defines the T-duality transformation.
\end{theorem}

\begin{theorem}[\cite{BEM04a}]\label{TD2}
Let $A$, $\widehat{A}$ denote connection forms on $Z$ and $\widehat{Z}$ respectively, choose an invariant representative $H \in [H]$ and $\widehat{H} \in [\widehat{H}]$, and let $\big(\Omega^{*}(Z)^{S^1}, d + H\big)$ denote the $H$-twisted, $\mathbb{Z}_2$-graded differential complex of invariant differential forms.

Then the following Hori map:
\begin{align*}
T: (\Omega^{*}(Z)^{S^1}, d-H) &\to (\Omega^{*+1}(\widehat{Z})^{\widehat{S}^1}, -(d+{\widehat{H})})\\
\omega \quad &\mapsto \int_{S^1}\omega \wedge e^{-\widehat{A} \wedge A},
\end{align*}
is a chain map isomorphism between the twisted, $\mathbb{Z}_2$-graded complexes. Furthermore, it induces an isomorphism on the twisted cohomology:
\begin{align*}
T: H^{*}_{d+H}(Z) &\to H^{*+1}_{d+\widehat{H}}(\widehat{Z}).
\end{align*}
\end{theorem}


\subsection{Graded T-duality for $2d$ $\sigma$-models} \label{graded}

The purpose of the subsection is to establish the graded T-duality with $H$-flux for $2d$ sigma models. 

In  \cite{HM21}, we introduced graded T-duality. Let us first recap the main results in \cite{HM21} using the language of sheaves, motivated by the elliptic sheaves introduced by Berwick-Evans in \cite{DB19}.

Let  $M$ be a smooth manifold carrying a gerbe with connection $(H, B_\alpha, F_{\alpha\beta}, (L_{\alpha\beta}, \nabla^{L_{\alpha\beta}}))$. 
\begin{definition} Define a sheaf $(\GG(M, H), D^H)$ on $\HH$ of commutative differential graded algebras that to $U\subset \HH$ assigns the graded complex of $\cO(U)$-modules
\be   (\GG(M, H)(U), D^H):= \bigoplus_{m\in \ZZ}\left(\cO(U; \Omega^{\bullet}(M)[[u, u^{-1}]])\cdot y^m, \, d+u^{-1}mH \right),  \ee
where $y$ is a variable to indicate the grading.  

\end{definition}

\begin{remark}{\em In \cite{HM21}, we constructed global sections of $(\GG(M, H), D^H)$, which have the Jacobi properties. 

Let $\Gamma$ be a subgroup of $SL(2, \ZZ)$ of finite index. Let $L$ be an integral lattice in $\CC$ preserved by $\Gamma$. Denote $\bH$ the upper half plane. A {\bf (meromorphic) Jacobi form} (c.f. \cite{EZ, Liu95}) of weight $s$ and index $l$ over $L\rtimes \Gamma$ is a (meromorphic) function $J(z, \tau)$ on $\CC \times \bH$ such that \newline
(i) $J\left(\frac{z}{c\tau+d}, \frac{a\tau+b}{c\tau+d}\right)=(c\tau+d)^se^{2\pi \sqrt{-1} l(cz^2/(c\tau+d))}J(z,\tau)$;\newline
(ii) $J(z+\lambda \tau+\mu, \tau)=e^{-2\pi \sqrt{-1}l(\lambda^2\tau+2\lambda z))}J(a, \tau),$
where 
$$(\lambda, \mu)\in L, \ \left(\begin{array}{cc} a&b\\ c&d\end{array}\right)\in \Gamma. $$
A slight extension of the above definition of Jacobi forms was adopted in \cite{HM21},  namely, (i) we will allow $J(z, \tau)$ to take values in the differential forms on a manifold $M$; (ii) as $J(z, \tau)$ takes values in differential forms, we don't require the singular points be poles but only remain undefined. 

A {\em gerbe module} over $(M, (H, B_\alpha, F_{\alpha\beta}, (L_{\alpha\beta}, \nabla^{L_{\alpha\beta}})))$ consists of the following data: \newline
(i) $E=\{E_{\alpha}\}$
is a collection of (infinite dimensional) separable Hilbert bundles $E_{\alpha}\to U_{\alpha}$ whose structure group is reduced to
$U_{\gI}$, which are unitary operators on the model Hilbert space $\gH$ of the form identity + trace class operator.
Here $\gI$ denotes the Lie algebra of  $U_{\gI}$, the trace class operators on $\gH$.\newline
(ii) on the overlaps $U_{\alpha\beta}$ 
there are isomorphisms
\beq
\phi_{\alpha\beta}: L_{\alpha\beta} \otimes E_\beta \cong E_\alpha,
\eeq
which are consistently defined on
triple overlaps because of the gerbe property. 

A {\em gerbe module connection} $\nabla^E$ is a collection of connections $\{\nabla^E_{\alpha}\}$ of the form $\nabla^E_{\alpha} = d + A_\alpha^E$, where $A_\alpha^E
\in \Omega^1(U_\alpha)\otimes \gI$ whose curvature $F^{E_\alpha}$ on the overlaps $U_{\alpha\beta}$ satisfies
\beq
\phi_{\alpha\beta}^{-1}(F^{E_\alpha}) \phi_{\alpha\beta} =  F^{L_{\alpha\beta}} I  +    F^{E_\beta}.
\eeq

Let $E=\{E_{\alpha}\}$ and $E'=\{E'_{\alpha}\}$ 
be {gerbe modules} for the gerbe $\{L_{\alpha\beta}\}$. Then an element of twisted K-theory $K^0(M, H)$
is represented by the pair $(E, E')$, see \cite{BCMMS}. Two such pairs $(E, E')$ and $(G, G')$ are equivalent
if $E\oplus G' \oplus K \cong E' \oplus G \oplus K$ as gerbe modules for some gerbe module $K$ for the gerbe $\{L_{\alpha\beta}\}$. 

Given a gerbe module pair $(E, E')$, motivated by the theory of elliptic genus, in \cite{HM21}, we constructed the Witten gerbe module pairs $\{\Theta_2(E), \Theta_2(E')\}$,
which are elements in $\oplus_{m\in \ZZ}K(M, mH)[q^{1/2}],$ where $q=e^{2\pi \sqrt{-1} \tau}$.  Assembling the {\it twisted Chern character} $Ch_H$ introduced in \cite{BCMMS, MS} for gerbe modules, in \cite{HM21}, we introduced the {\it graded twisted Chern character} $\mathrm{GCh}_H$ and showed that when the degree 4 component $Ch_H^{[4]}(E, E')$ is vanishing, 
\be \gch\left(\frac{\Theta_2(E)}{\Theta_2(E')}\right)=\det\left(\frac{\theta_2(z+u^{-1}(B+F^{E}), \tau)}{\theta_2(z+u^{-1}(B+F^{E'}), \tau)}\right)\in \bigoplus_{m\in \ZZ}\Omega^{ev}(M)_{(d+mH)-cl}[[q^{1/2}, u, u^{-1}]]\cdot y^m\ee when $u=1$, is the expansion at $y=0$ of a Jacobi form of index 0 of the two variables $(\tau, z)$ (with $y=e^{2\pi \sqrt{-1} z}$) over $\ZZ^2\rtimes \Gamma^0(2)$, where $\Gamma^0(2)$ is an index 2 subgroup of $SL(2, \ZZ)$. Therefore the expansion of $\gch\left(\frac{\Theta_2(E)}{\Theta_2(E')}\right)$ at $y=0$ gives a global section of the sheaf  $(\GG(M, H), D^H)$ with Jacobi property. For details as well as other types of Witten gerbe modules, we refer to \cite{HM21}. }  
\end{remark}

Back to the situation of T-duality, we will require our data to respect the $\TT$ and $\widehat \TT$ actions on the fiber directions.  

\begin{definition} For the pair $(Z, H)$, define a sheaf $(\GG(Z, H), D^H)$ on $\HH$ of commutative differential graded algebras that to $U\subset \HH$ assigns the graded complex of $\cO(U)$-modules
\be   (\GG(Z, H)(U), D^H):= \bigoplus_{m\in \ZZ}\left(\cO(U; \Omega^{\bullet, \TT}(Z)[[u, u^{-1}]])\cdot y^m, \, d+u^{-1}mH \right).    \ee
Dually, we can define the sheaf $(\GG(\widehat Z, \widehat H), D^{\widehat H}).$ Passing to cohomology, we get the sheave $\GGG(Z, H)$ more precisely, define $\GGG(Z, H)$ by
\be \GGG(Z, H)(U):= \bigoplus_{m\in \ZZ}\cO(U; H(\Omega^{\bullet, \TT}(Z)[[u, u^{-1}]], d+u^{-1}mH))\cdot y^m.\ee
Similarly one has $\GGG(\widehat Z, \widehat H)$ on the dual side. 
\end{definition}

One can define the {\bf graded Hori morphisms} as follows. For $m\in \ZZ$, define the {\bf level $m$ Hori map} by
\be T_{*, m}(G)= \int_\TT e^{-mA\wedge \widehat A } G,\ee
for $G$ is an $\TT$-invariant form on $Z$ and $(d-mH)G=0$. As we have 
\be m\widehat H=mH+d(mA\wedge \widehat A),\ee it is not hard to see that $T_{*, m}G$ is a $\widehat \TT$-invariant form on $\widehat Z$ and $$(d+m\widehat H)(T_{*, m}(G))=0,$$ similar to the $m=1$ case. 
Define the graded Hori map
\be GHor_*(U):  \bigoplus_{m\in \ZZ}\cO(U; \Omega^{\bar k, \TT}(Z)[[u, u^{-1}]])\cdot y^m\to \bigoplus_{m\in \ZZ}\cO(U; \Omega^{\overline{k+1}, \widehat \TT}(\widehat Z)[[u, u^{-1}]])\cdot y^m \ee
by
\be  GHor_*(U)\left(\sum_{m\in \ZZ}\omega_m y^m\right)=\sum_{m\in \ZZ} T_{*, m}(\omega_m)y^m,\ee
where
$\sum_{m\in \ZZ}\omega_m\cdot y^m$, with $\omega_m\in \cO(U; \Omega^{\bar k, \TT}(Z)[[u, u^{-1}]])$ and $\bar k$ denotes the parity of $k$. 
It is not hard to see that $GHor_*(U)$ is a chain map. Therefore $\{GHor_*(U)\}$ gives a sheave morphism
\be GHor_*: (\GG(Z, H), D^H)\to  (\GG(\widehat Z, \widehat H), D^{\widehat H}), \ee
which we call {\bf graded Hori morphism}. Passing to cohomology, we have the graded Hori morphism
\be GHor: \GGG(Z, H) \to \GGG(\widehat{Z}, \widehat H). \ee

One can similarly define the graded Hori morphism on the dual side,
\be \widehat{GHor}_*:  (\GG(\widehat Z, \widehat H), D^{\widehat H}) \to  (\GG(Z, H), D^H), \ee
and 
\be \widehat{GHor}: \GGG(\widehat Z, \widehat H) \to \GGG(Z, H).\ee

The main result in \cite{HM21} can be stated as the following {\bf graded T-duality with $H$-flux } theorem,
\begin{theorem} \label{GT}
\be \widehat{GHor}_* \circ GHor_*=-y\frac{\partial }{\partial y}, \ \  GHor_*\circ \widehat{GHor}_* =-y\frac{\partial }{\partial y}.\ee
\be \widehat{GHor} \circ GHor=-y\frac{\partial }{\partial y}, \ \  GHor\circ \widehat{GHor} =-y\frac{\partial }{\partial y}.\ee
\end{theorem}

Using double loop spaces, we also construct the following sheaves, 
\begin{definition} For the pair $(Z, H)$, let $\llz^H=LLZ^H$ be the total space of the circle bundle over $\llz$ as defined in Section \ref{construction} and define a sheaf $(\GG(\llz^H, \cL_{H}), \mathcal{Q}_H)$ on $\HH$ of commutative differential graded algebras that to $U\subset \HH$ assigns the graded complex of $\cO(U)$-modules
\be
(\GG(\llz^H, \cL_H)(U), \mathcal{Q}_H):= \bigoplus_{m\in \ZZ}\left(\cO(U; \Omega^{\bullet, \TT}_{bas}(\llz^H, \cL_{H, \tau}^{\otimes m})[[u, u^{-1}]]^{K_1+\tau K_2})\cdot y^m,\  Q_{mH,\tau}\right), \\
 \ee
where \be Q_{mH,\tau}:=\nabla^{\cL_{H, \tau}^{\otimes m}}-u\iota_{K_1+\tau K_2}+u^{-1}m\widetilde{p}^*\overline{\overline{H}}.  \ee

Dually, one can define the sheaf $ (\GG(\llhatz^{\widehat H}, \cL_{\widehat H}), \mathcal{Q}_{\widehat H})$. Passing to cohomology, we get the sheaves
$\GGG(\llz^H, \cL_H)$ and $\GGG(\llhatz^{\widehat H}, \cL_{\widehat H})$. 
\end{definition} 

By Theorem \ref{localization} (using the $\TT$ and $\widehat \TT$ invariant versions), we have 
\begin{theorem} The restriction maps 
\be res: \GGG(\llz^H, \cL_H)\to \GGG(Z, H), \ \ \widehat{res}:  \GGG(\llhatz^{\widehat H}, \cL_{\widehat H})\to \GGG(\widehat Z, \widehat H)\ee
are isomorphisms of sheaves. 

\end{theorem}

Denote by 
\be 
\begin{split}
&GHor^\sigma:= \widehat{res}^{-1}\circ GHor \circ res: \GGG(\llz^H, \cL_H)\to\GGG(\llhatz^{\widehat H}, \cL_{\widehat H}), \\
&\widehat{GHor}^\sigma:=res^{-1}\circ \widehat{GHor}\circ \widehat{res}: \GGG(\llhatz^{\widehat H}, \cL_{\widehat H}) \to \GGG(\llz^H, \cL_H),
\end{split}
\ee
which we call {\bf graded Hori morphisms for $2d$ $\sigma$-models}.
We assemble them into the following commutative diagram,
\[ 
\begin{tikzcd}
\GGG(\llz^H, \cL_H) \ar[rr, shift left]{rr}{GHor^\sigma} \ar[dd]{dd}{res\, \cong } & &\GGG(\llhatz^{\widehat H}, \cL_{\widehat H})\ar[ll, shift left]{ll}{ \widehat{GHor}^\sigma} \ar[dd]{dd}{\widehat{res}\, \cong }
\\
& &\\
\GGG(Z, H)\ar[rr, shift left]{rr}{GHor} & & \GGG(\widehat Z, \widehat H)
\ar[ll, shift left]{ll}{ \widehat{GHor}}
\end{tikzcd}
\]

By Theorem \ref{GT}, we obtain the following {\bf graded T-duality with $H$-flux for $2d$ $\sigma$-models}
\begin{theorem}\label{main}
\be \widehat{GHor}^\sigma \circ GHor^\sigma=-y\frac{\partial }{\partial y}, \ \  GHor^\sigma\circ \widehat{GHor}^\sigma =-y\frac{\partial }{\partial y}.\ee
\end{theorem}


\end{document}